\documentclass[12pt]{amsart}

\usepackage[top=1in, bottom=1in, left=1in, right=1in]{geometry}

\usepackage{color}
\usepackage{amssymb}
\usepackage{comment}
\usepackage{pdfpages}
\usepackage{graphicx}
\usepackage{dcolumn}

\RequirePackage[numbers]{natbib}
\RequirePackage[colorlinks=true, pdfstartview=FitV, linkcolor=blue,
  citecolor=blue, urlcolor=blue]{hyperref}
\usepackage[dvipsnames]{xcolor}

\usepackage{graphicx}
\graphicspath{{./figures/}}
\usepackage{amsfonts}
\usepackage{amsmath}
\usepackage{amsthm}
\usepackage{amssymb}
\usepackage{amsbsy}
\usepackage{epsfig}
\usepackage{fullpage}
\usepackage{natbib, mathrsfs} 
\usepackage{verbatim}
\usepackage[latin1]{inputenc}
\usepackage{mhequ}

\usepackage{caption}
\usepackage{subcaption}
\usepackage{multirow}
\usepackage{multicol}
\usepackage{array}

\numberwithin{equation}{section}

\newcommand{\VV}{\bar{V}}
\newcommand{\TT}{\bar{T}}
\newcommand{\JJ}{\bar{J}}
\newcommand{\xx}{\bar{x}}
\newcommand{\R}{\mathbb{R}}

\newcommand{\be}{\begin{equs}}
\newcommand{\ee}{\end{equs}}

\newcommand{\N}{\mathbb{N}}

\usepackage{bbm} 

\newcommand{\Pos}{\mathbf{P}}
\newcommand{\Posref}{\mathbf{P}_{\mathrm{bounce}}}
\newcommand{\Mom}{\mathbf{M}}
\newcommand{\Rot}{R}
\newcommand{\Rotref}{R_{\mathrm{bounce}}}
\newcommand{\Th}{\Theta}
\newcommand{\K}{K}

\newtheorem{theorem}{Theorem}[section]

\newtheorem{remark}[theorem]{Remark}

\newtheorem{prop}[theorem]{Proposition}

\begin{document}

\title{Reversible and non-reversible Markov Chain Monte Carlo  algorithms for reservoir simulation  problems}

\author{ P.Dobson, I. Fursov,  G. Lord and M. Ottobre}

\address{\noindent \textsc{Paul Dobson, Maxwell Institute for Mathematical Sciences,  Department of Mathematics, Heriot-Watt University, Edinburgh EH14 4AS, UK}} 
\email{pd14@hw.ac.uk}
\address{\noindent \textsc{Ilya Fursov,   Institute of Petroleum Engineering, Heriot-Watt University, Edinburgh EH14 4AS, UK}} 
\email{i.fursov@hw.ac.uk}
\address{\noindent \textsc{Gabriel Lord, Maxwell Institute for Mathematical Sciences,  Department of Mathematics, Heriot-Watt University, Edinburgh EH14 4AS, UK}} 
\email{g.j.lord@hw.ac.uk}
\address{\noindent \textsc{Michela Ottobre, Maxwell Institute for Mathematical Sciences, Department of Mathematics, Heriot-Watt University, Edinburgh EH14 4AS, UK}} 
\email{m.ottobre@hw.ac.uk}

\begin{abstract}
We compare numerically the performance of reversible and non-reversible Markov Chain Monte Carlo algorithms for high dimensional oil reservoir problems;   because of the nature of the problem at hand, the target  measures from which we sample are supported on bounded domains. We compare two strategies to deal with bounded domains, namely reflecting proposals off the boundary and rejecting them when they fall outside of the domain.   
We observe that for complex high dimensional problems reflection mechanisms outperform rejection approaches and that the advantage of introducing non-reversibility in the Markov Chain employed for sampling is more and more visible as the dimension of the parameter space increases.  
 \vspace{5pt}
\\
{\sc Keywords.} {Markov Chain Monte Carlo methods, Non-reversible Markov Chains, Subsurface Reservoir Simulation, High-dimensional Sampling. }

\end{abstract}

\date{\today}

\maketitle

\section{Introduction}

 Markov Chain Monte Carlo (MCMC) methods are popular algorithms which allow one to sample from a given target measure $\pi$ on $\R^N$. In combination with the Bayesian inference approach, MCMC methods have been very successfully implemented in a vast range of problems in the applied sciences,   and the literature about MCMC is extensive. The purpose of  MCMC algorithms is to build a Markov chain $\{x_k\}_{k \in \N}$ which has the measure $\pi$ as invariant measure. Traditionally this is obtained by  ensuring that the chain satisfies the {\em detailed balance} condition with respect to the measure $\pi$, so that the resulting chains are {\em reversible} with respect to $\pi$. In recent years,  {\em non-reversible} MCMC algorithms have attracted a lot of attention, because of their favourable convergence and mixing properties; the literature on the matter has rapidly become large, here we refer the reader to  e.g. \cite{cit1,cit2,cit3,cit4, {Ottirr}} and references therein;  however, to the best of our knowledge most of the papers on non-reversible MCMC so far have tested this new class of algorithms only on relatively simple target measures. Furthermore, the performance of non-reversible algorithms has been discussed almost exclusively in the case in which the measure is supported on the whole of $\R^N$. 
However in many applications it is very important to be able to sample from measures supported on bounded domains. This is the case, for example, in applications to reservoir modelling and petroleum engineering, which we treat in this paper. The purpose of this paper is twofold: on the one hand we want to test the performance of non-reversible algorithms for complex, high-dimensional problems, which are completely out of reach for a full analytical treatment; on the other hand, we want to employ them for situations in which the target measure is supported in a bounded domain. 
The non-reversible algorithms that we consider in this paper are the so-called Horowitz algorithm, see \cite{hor:91},  and the Second Order Langevin- Hamiltonian Monte Carlo (SOL-HMC) algorithm, introduced in \cite{{OPSP}}. Both of them are non-reversible modifications of the well known Hamiltonian Monte Carlo (HMC)  \cite{RadfordNeal}, which is reversible. More precisely, the Horowitz algorithm is a non-reversible version of HMC and the SOL-HMC algorithm is a modification of the Horowitz algorithm, well-posed in infinite dimensions and therefore well-adapted to  sample from the high-dimensional target measures that we will treat here. 

All the algorithms we discuss in this paper need in principle no modification in order to sample from  measures supported on bounded domains.    However, if they are not suitably modified, they will employ proposal moves which fall outside of the support of the target measure. For the problem we consider, this seems to give two major drawbacks, namely i) proposal moves that fall outside of the box are immediately rejected, so the algorithm wastes time rejecting moves which one knows a priori should not be made;  \footnote{Admittedly, this observation only applies when the size of the domain is known a priori. See also \cite{Duncan}.}  ii) the likelihood function is calculated through the use of a simulator, which, further to being time-consuming to run, it will mostly fail to calculate correctly values that fall outside the support of the target. For this reason, we will consider two modifications of each one of the mentioned algorithms in which proposal moves that fall outside of the support of the target measures are either rejected or bounced off (or better, reflected off) the boundary of the support (see Section \ref{sec2}), so that the proposal is not automatically rejected as it will fall within the support.   With these observations in mind, let us come to summarize the main conclusions of the paper:
\begin{itemize}
    \item We compare rejection and reflection strategies and test them on both low and high dimensional targets and conclude that, for the problems at hand, the two strategies perform similarly when implemented in low dimensions; however in high dimensions (and for more complex problems where a proxy is employed for the likelyhood function), reflections seem more advantageous
    \item we compare the performance of HMC, Horowitz and SOL-HMC and conclude that, in high dimensions, the SOL-HMC algorithm is subtantially outperforming the other two. 
\end{itemize}
Performance of all the algorithms is compared by using the normalized Effective Sample size (nESS) as a criterion for efficiency, see Section \ref{sec4}. 
We emphasize that one of the main purposes of this paper is to demonstrate how the SOL-HMC algorithm, while being a simple modification of the HMC algorithm, which requires truly minimal code adjustment with respect to HMC,  can bring noticeable improvements with respect to the latter method; furthermore, such improvements are more noticeable when tackling high-dimensional complex target measures.  

The paper is organised as follows: in Section \ref{sec2} we recall the HMC, SOL-HMC and Horowitz algorithms, present the numerical integrators that we use in order to implement such algorithms and introduce the versions of such methods which are adapted to sampling from measures with bounded support. In Section \ref{sec3} we give details about the types of target measures used to compare the efficiency of these various algorithm and how they arise from  reservoir simulation problems. This section explains mostly the mathematical structure of such target measures. Further details regarding the simulator and some basic background material about the reservoir model are deferred to Appendix B. In  Section \ref{sec4} we present numerical experiments. For completeness, we include Appendix A, containing some simple theoretical results regarding the modified algorithms.

\section{Description of the algorithms}\label{sec2}
In this section we present  the three main algorithms that we would like to compare, namely the Hamiltonian Monte Carlo (HMC) algorithm, the SOL-HMC algorithm and the Horowitz algorithm.  With abuse of notation, throughout we will denote a probability measure and its density with the same letter, i.e. $\pi(dx)= \pi(x)dx$. 

 Suppose we wish to sample from a probability measure $\pi$ defined on $\R^N$ which has a density of the form
\be
\pi(x) \propto e^{-V(x)}e^{-\langle x, C^{-1}x\rangle}, 
\ee
i.e. the target density $\pi$ is absolutely continuous with respect to a Gaussian measure with covariance matrix $C$ 
(as customary, we assume that such a matrix is symmetric and positive definite). 
All three of our algorithms make use of the common trick of introducing an auxiliary variable $p\in \R^N$ and sampling from the density $\tilde{\pi}$ defined on the  extended state space $\R^N \times \R^N$ as follows
\be\label{auxtarget}
\tilde{\pi}(x,p) \propto e^{-V(x)}e^{-\langle x, C^{-1}x\rangle}e^{-\frac{1}{2}p^2}.
\ee
The original target measure  $\pi$ is therefore the marginal of $\tilde{\pi}$ with respect to the variable $x$. More precisely,  the algorithms we will consider  generate a chains $\{(x^k, p^k)\}_k \subset \R^N\times \R^N$ which sample from the measure  \eqref{auxtarget}; because \eqref{auxtarget} is a product measure of our desired target times a standard Gaussian, if we consider just the first component of the chain $\{(x^k, p^k)\}_k$, i.e. the chain $\{x^k\}_k$, then, for $k$ large enough, such a chain will be sampling from the measure $\pi$. We now focus on explaining how the chain $\{(x^k, p^k)\} \subset \R^N\times \R^N$ is generated by the three algorithms we consider. 

Let us  introduce the Hamiltonian function
\be\label{eq:Hamiltonianfunct}
H(x,p) = V(x)+\langle x, C^{-1}x\rangle+\frac{1}{2}p^2 \,;
\ee
then the associated Hamiltonian flow can be written as 
\begin{equation}\label{eq:Hamiltonianflow}
    \begin{cases}
    \dot{x} = p\\
    \dot{p} = -x- C\nabla V(x).
    \end{cases}
\end{equation}

Let $\chi^t$ denote a numerical integrator for the system \eqref{eq:Hamiltonianflow} up to time $t$ (we will comment below on our choices of integrator). The {\bf HMC algorithm} then proceeds as follows:  suppose that at time $k$ the first component of the chain is $x^k$ and 
\begin{enumerate}
    \item pick $p^k \sim N(0,I)$;
    \item compute
    $$
    (\tilde{x}^{k+1},\tilde{p}^{k+1})=\chi^\delta(x^k,p^k)
    $$
    and propose $\tilde{x}^{k+1}$ as the next move;
    \item calculate the  acceptance probability $\alpha_k$, according to
    \be
    \alpha_k = \min(1, e^{-(H(\tilde{x}^{k+1},\tilde{p}^{k+1})-H(x^k,p^k))};
    \ee
    \item Set $q^{k+1}=\tilde{q}^{k+1}$ with probability $\alpha_k$, otherwise set $q^{k+1}=q^k$.
\end{enumerate}

In principle any numerical integrator can be used in an HMC algorithm (see \cite{Sanz, RadfordNeal} for more detailed comments on this). In this paper we will consider two numerical integrators $\chi$, which are two popular splitting schemes.  The first is given by splitting the ``momentum'' and ``position'' equations, see e.g. \cite{Sanz} and references therein. That is,  let  $\Mom^t$ denote the solution map at time $t$ of the system 
\begin{equation}
    \begin{cases}
    \dot{x} = 0,\\
    \dot{p} = -x-C\nabla V(x)
    \end{cases}
\end{equation}
and $\Pos^t$ denote the solution map at time $t$ of the system 
\begin{equation}
    \begin{cases}
    \dot{x} = p\\
    \dot{p} = 0.
    \end{cases}
\end{equation}
For HMC we shall always use the numerical integrator \begin{equation}\label{intHMC}
\chi^{\delta}_H = \Mom^{\delta/2} \Pos^\delta \Mom^{\delta/2} \,.
\end{equation}
Note that we can always write the maps $\Mom^t$  and $\Pos^t$ explicitly;  indeed,
\begin{align}\label{pdelta}
    \Mom^{\delta/2}(x,p) &= \left(x, p-\frac{\delta}{2}x- C\nabla V(x)\right)\\
    \Pos^\delta(x,p) &=\left(x+\delta p,p\right) \,.
\end{align}

The other splitting scheme that we will consider splits  the Hamiltonian system \eqref{eq:Hamiltonianflow} into its linear and nonlinear part. More precisely, let $\Rot^t$ and $\Th^t$ be the flows associated with the following ODEs:
\begin{equation}
    \Rot^t:\begin{cases}
        \dot{x} = p,\\
    \dot{p} = -x,
    \end{cases}
    \Th^t:\begin{cases}
        \dot{x} = 0,\\
    \dot{p} = - C \nabla V(x).
    \end{cases}
\end{equation}
  The resulting integrator is given by 
 \be\label{intSOL}
  \chi^{\delta}_{S} = \Th^{\delta/2} \Rot^{\delta} \Th^{\delta/2}. 
  \ee
This is the integrator that we will use in the SOL-HMC algorithm (see step (1) of the SOL-HMC algorithm below); the use of this splitting scheme for high dimensional problems has been studied in \cite{InfHMC}.  SOL-HMC  is motivated as a time-discretisation of the SDE
\begin{align}\label{SOLSDE}
    \begin{cases}
    dx&=p dt,\\
    dp&=[-x- C\nabla V(x)]dt-pdt+\sqrt{2 C}dW_t,
    \end{cases}
\end{align}
where $\{W_t\}_{t\geq 0}$ is a standard $N$-dimensional Brownian motion. Such an equation can be seen as a Hamiltonian dynamics perturbed by an Ornstein-Uhlenbeck process in the momentum variable. As is well known, the SDE \eqref{SOLSDE} admits the measure \eqref{auxtarget} as unique invariant measure, see e.g.  \cite{Villani}.  With these observations in mind, define $\mathcal{O}^\varepsilon$ to be the map which gives the solution at time $\varepsilon$ of the system
\begin{align}\label{oe1}
\begin{cases}
    dx&=0,\\
    dp&=-pdt+\sqrt{2}dW_t.
\end{cases}
\end{align}
Note that we may solve this system explicitly, indeed
\begin{equation} \label{oe2}
\mathcal{O}^\varepsilon(x,p) = (x, pe^{-\varepsilon}+(1-e^{-2\varepsilon})^{\frac{1}{2}}\xi)
\end{equation}
where $\xi$ is a standard normal random variable. In Section \ref{sec4} we will set  
\begin{equation}\label{parami}
    e^{-2\varepsilon}=1-i^2, 
\end{equation}    
where $i$ is a parameter we can tune;  in which case we have
\be
\mathcal{O}^\varepsilon(x,p) = (x, pe^{-\varepsilon}+i\xi).
\ee
The {\bf SOL-HMC algorithm} is as follows:
\begin{enumerate}
    \item Given $(x^k,p^k)$, let
    \be
    (\hat{x}^k,\hat{p}^k) = \mathcal{O}^\varepsilon(x^k,p^k)
    \ee
    and propose 
    \be
    (\tilde{x}^{k+1},\tilde{p}^{k+1}) = \chi^\delta_S(\hat{x}^k,\hat{p}^k), 
    \ee
    where we recall that $\chi^{\delta}_S$ is the integrator introduced in \eqref{intSOL}; 
    \item calculate the acceptance probability $\alpha_k$ according to
    \be\label{eq:accrejprob}
    \alpha_k = \min(1, e^{-(H(\tilde{x}^{k+1},\tilde{p}^{k+1})-H(\hat{x}^k,\hat{p}^k))});
    \ee
    \item set
    \be
    (x^{k+1},p^{k+1}) = \begin{cases}
    (\tilde{x}^{k+1},\tilde{p}^{k+1}) & \text{ with probability } \alpha_k,\\
    (\hat{x}^k,-\hat{p}^k) & \text{ with probability } 1-\alpha_k.
    \end{cases}
    \ee
\end{enumerate}
Finally, the algorithm that we will refer to as the {\bf Horowitz algorithm}, is just the SOL-HMC algorithm when in step one, instead of using the integrator $\chi_S$, we use the integrator $\chi_H$ (defined in \eqref{intHMC}). 

\begin{remark}\label{rem31}\textup{
We do not give many details about HMC, SOL-HMC and the Horowitz algorithm here, and refer to the already cited literature. However we would like to stress the two following facts: 
\begin{itemize}
\item The chain $\{x^k\}_k$ produced by the HMC algorithm is reversible with respect to the measure $\pi$, in the sense that it satisfies detailed balance with respect to $\pi$ \cite{RadfordNeal} -- more precisely, the chain $\{(x^k,p^k)\}_k$ generated by HMC satisfies a generalised detailed balance condition with respect to $\tilde{\pi}$, see e.g. \cite[Lemma 1]{Sanz} or \cite{Lelievre}.  In contrast the chains generated by the Horowitz algorithm and the SOL-HMC do not satisfy any form of detailed balance with respect to $\tilde{\pi}$ and they are therefore non-reversible, see \cite{Ottirr, {OPSP}}.  In Appendix A we will show that adding reflections to the algorithm does not alter this property. 
\item The difference between the Horowitz algorithm and HMC may seem small, but in reality this small difference is crucial. Indeed, thanks to this choice of integrator, SOL-HMC is well-posed in infinite dimensions, while the Horowitz algorithm is not. For a discussion around this matter see \cite{Ottirr, InfHMC}. 
\end{itemize}
}
\end{remark}

As mentioned in the introduction, in this paper we will be interested in sampling from measures which are not necessarily supported on the whole space $\R^N$, but just on     some box $B=[-a,a]^N$.  If this is the case then one may still use any one of the above algorithms and reject  proposal moves that fall outside the box. We will briefly numerically analyse this possibility (see Section \ref{sec4}). Alternatively, one may want to simply make sure that all the proposed moves belong to the box $B$, so that the algorithm doesn't waste too much time rejecting the moves that fall outside the box.  We therefore consider modified versions of the introduced algorithms by  introducing appropriate reflections to ensure that all of the proposals belong to the box $B$. Because the proposal is defined through numerical integration of the Hamiltonian dynamics, we will need to modify the integrators $\chi_H$ and $\chi_S$. 

First consider the map $\Pos^\delta$ defined in \eqref{pdelta}; then we define map $\Posref^\delta$ recursively as follows:
\begin{enumerate}
    \item If $\Pos^{\delta}(x,p)\in B$ then set $\Posref^\delta(x,p)=\Pos^\delta(x,p)$.
    \item Otherwise define 
    \begin{equation}
        \alpha = \inf\{\beta\in[0,1]: \Pos^{\beta \delta}(x,p)\notin B\}.
    \end{equation}
    In which case $\Pos^{\alpha \delta}(x,p)$ lies on the boundary of the box, so there exists some\footnote{It could occur that there is more than one $j$ such that the $j$th component of $\Pos^{\alpha\delta}(x,p)$ is $\pm a$, in which case apply the operator $S_j$ for all such $j$.} $j\in\{1,\ldots, N\}$ such that the $j$th component of $\Pos^{\alpha\delta}(x,p)$ is either $a$ or $-a$. Then we define 
    \begin{equation}
        \Posref^\delta(x,p) = \Posref^{(1-\alpha)\delta}(S_j(\Pos^{\alpha \delta}(x,p)).
    \end{equation}
    Here $S_j$ is the reflection map $S_j(x,p)=(x,p_1,\ldots,p_{j-1},-p_j,p_{j+1},\ldots,p_N)$.
\end{enumerate}

Similarly we define $\Rotref^\delta$ by
\begin{enumerate}
    \item If $\Rot^{\delta}(x,p)\in B$ then set $\Rotref^\delta(x,p)=\Rot^\delta(x,p)$.
    \item Otherwise define 
    \begin{equation}
        \alpha = \inf\{\beta\in[0,1]: \Rot^{\beta \delta}(x,p)\notin B\}.
    \end{equation}
    In which case $\Rot^{\alpha \delta}(x,p)$ lies on the boundary of the box, so there exists some $j\in\{1,\ldots, N\}$ such that the $j$th component of $\Rot^{\alpha\delta}(x,p)$ is either $a$ or $-a$. Then we define 
    \begin{equation}
        \Rotref^\delta(x,p) = \Rotref^{(1-\alpha)\delta}(S_j(\Rot^{\alpha \delta}(x,p))).
    \end{equation}
\end{enumerate}
Note that it may occur that $\Rot^{\delta}(x,p)\in B$ however there is some point $\alpha\in[0,1]$ such that $\Rot^{\alpha\delta}(x,p)\notin B$, in this case we still set $\Rotref^{\delta}(x,p)=\Rot^{\delta}(x,p)$. Therefore the algorithm {\textit {\bf HMC-bounce}} ({\textit {\bf Horowitz-bounce}}, respectively) is defined like   HMC (Horowitz, respectively), but the numerical  $\chi^{\delta}_{H, Bounce} = \Mom^{\delta/2} \Posref^\delta \Mom^{\delta/2}$ is employed in place of the integrator $\chi_H$; analogously,  the algorithm {\em {\bf SOL-HMC-bounce}} is defined as the algorithm SOL-HMC with numerical integrator $\chi^{\delta}_{S,Bounce} = \Th^{\delta/2} \Rotref^{\delta} \Th^{\delta/2}$ in place of $\chi_S$. 

%

\section{Target measures}\label{sec3}
In this section we describe the three target measures that will be the object of our simulations. The first measure we consider, $\pi_{Ros}$,  is a change of measure from the popular 5D Rosenbrock, see \eqref{5DRos}.  This is the most artificial example we consider. The other two target measures are posterior measures for parameter values in reservoir simulation models. Roughly speaking, the second target measure we describe, $\pi_{full}$,  is a posterior measure on the full set of parameters of interest in a reservoir model; for our reservoir model, which is quite realistic, we will be sampling from 338 parameters, hence this measure will be a measure supported on $\R^{338}$. The third measure, $\pi_{light}$,  is a measure on $\R^{21}$, which derives from considering a considerably simplified reservoir model. We will refer to the former as {\em full reservoir model} and to the latter as {\em lightweight parametrization}. In this section we explain the mathematical structure of $\pi_{full}$ and $\pi_{light}$, without giving many details regarding the inverse problem related to the reservoir model. More details about the reservoir model and the simulator used to produce the likelihood function have been included in Appendix B for completeness. 
In the following we let $I_N$ denote the $N\times N$ identity matrix.
Let us now  come to describe our targets. 

\smallskip
$\bullet$ {\bf Change of measure from 5D Rosenbrock (i.e. $\pi_{Ros}$)}. The first target measure we consider is a measure on $\R^5$ and it is a change of measure from  the 5D Rosenbrock measure; namely, the density of 5D Rosenbrock is given by
\begin{align}
f(x)=\sum_{i=1}^4 100(x_{i+1}-x_i^2)^2 + (1-x_i)^2, \quad x=(x_1, {\dots}, x_5). 
\label{eq:ilya_1}
\end{align}
The target we consider is given by 
\begin{equation}\label{5DRos}
\pi_{Ros}(x) \propto e^{-\frac{1}{2} f(x)}e^{-\frac{1}{2} \langle x, C^{-1}x\rangle },
\end{equation}
where \(C\) is the prior covariance matrix. In all the numerical experiments (see Section \ref{sec4}) regarding $\pi_{Ros}$ we take $C=0.3\cdot I_{5}$.

$\bullet$ {\bf Full reservoir simulation}. We study a Bayesian inverse problem for reservoir simulation. 
We consider a synthetic reservoir with 7 layers, and 40 producing/injecting wells. A diagrammatic representation of the reservoir is shown in Figure~\ref{fig:ilya_1}.  Each layer is divided in blocks and, while the well goes through all the layers, it will not necessarily communicate through perforations with all the blocks it goes through (in Figure~\ref{fig:ilya_1} we highlight in yellow the boxes containing a perforation of a well). We also assumed that, in each layer, the well goes through at most one block. 
In total our subsurface reservoir is made of 124 blocks: 38  blocks on the boundaries to represent the active aquifers, one block per layer per well,  plus some additional blocks (which are neither acquifer blocks nor crossed by the wells).  

\begin{figure}[h]
\noindent
\includegraphics[keepaspectratio=true, width=\textwidth]{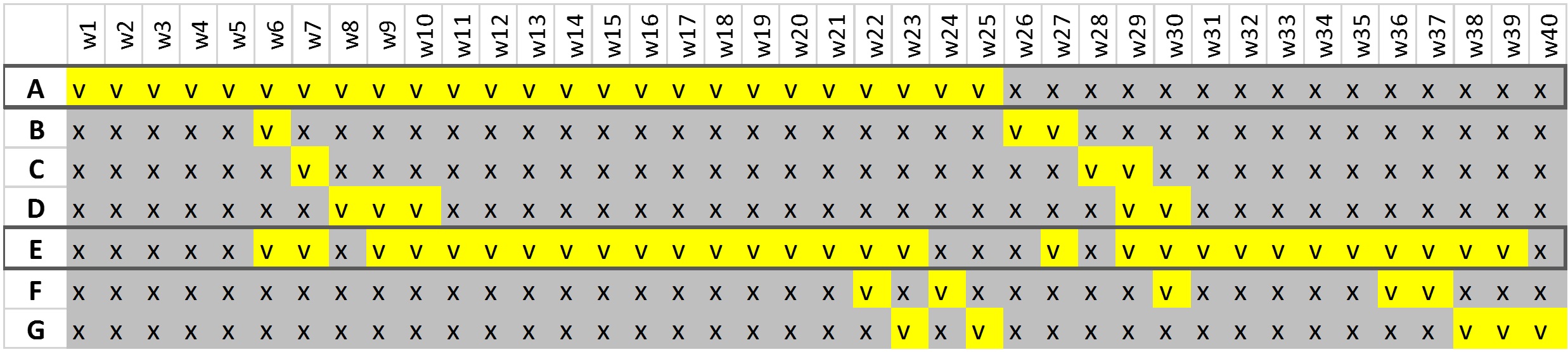}
\caption{Perforations of the 40 wells (\textit{columns}) in the seven layers (\textit{rows}). Yellow ``v" stands for the block containing a perforation of a well. That is, the well goes through all the layers, but there is a hole for well-block communication only in correspondence of the yellow boxes.  This figure does not show all the blocks -- but only those perforated by the wells. In particular, it does not show the aquifer blocks located on the boundary.}
\label{fig:ilya_1}
\end{figure}
The reservoir properties (i.e. the parameters that we will be interested in sampling from) are described by the \textit{pore volumes} \(V_\ell\) of the blocks, $\ell\in\{1, \dots, 124\}$,  the \textit{transmissibilities} \(T_{\ell j}\) between the interconnected blocks $\ell$ and $j$, and the \textit{perforation productivity coefficients} \(J_{w\ell}\) for the well-block connections. We do not explain here the practical significance of such parameters and for more details on reservoir simulation we refer the reader to \cite{Manual}.  
Altogether the parameter space for this example is 338-dimensional. 
For the sake of clarity all (nonzero) $T_{\ell j}$ are re-indexed with a single index as  $T_p$, $p \in \{1, \dots, 139\}$; and similarly $J_{w \ell}$ are re-indexed as $J_k$,
$k \in \{1\dots 75\}$ 
and we denote by $x \in \R^{338}$ the full vector  of parameters, i.e. $x=(V_1,\dots, V_{124}, T_1,  \dots,T_{139}, J_1, \dots J_{75})^T$. There are 86 non-aquifer blocks in total, and we always assume an ordering of the parameters $V_{\ell}$ such that the first 86 of them correspond to the non-aquifer blocks. 
In our   Bayesian inverse problem for the parameters $x$  the likelihood function is defined from the reservoir simulation output, and the prior is a Gaussian with covariance matrix $C$. The observed block pressure and the well bottom hole pressure (BHP) data are known for certain wells and time steps; we  arrange such data into the vector \(d_0\).  The likelihood \(L(d_0|x)\), see equation \eqref{eq:ilya_2} below, is defined using the simulator-modelled data \(d(x)\), the observed data \(d_0\), and the covariance matrix for data errors \(C_d\). The function $d(x)$ is found by numerical solution of a system of ordinary differential equations, which we report in Appendix B, see \eqref{eq:ilya_sim1} -- \eqref{eq:ilya_sim4};  such a system describes the relation between the vector of reservoir properties $x$ and the simulated pressures. The important thing for the time being is that such a system is high dimensional and the resulting posterior is analytically intractable.\footnote{The simulator we use also allows for fast calculation of the gradients of the log likelihood by the adjoint procedure \cite{Adjoint}, so that HMC-type samplers can be run. 
}
Finally, we seek to sample from the measure 
$$
\pi_{full} (x \vert d_0) \propto L(d_0 \vert x) \cdot e^{-\langle x, C^{-1}x \rangle},  
$$
where the likelihood function is of the form 
\begin{align}
L(d_0|x)=\exp\left(-\frac{1}{2}(d(x)-d_0)^T C_d^{-1} (d(x)-d_0)\right) \,.
\label{eq:ilya_2}
\end{align}
In our numerical experiments we will always take the matrix $C_d$ to be diagonal,  with the entries equal to either \(\sigma_{BHP}^2 = 20^2\) or \(\sigma_{b}^2 =9\).
We will give more details about this choice  in Appendix B.  
The full parameterisation is further divided into three subcases denoted here as \textit{full-a}, \textit{full-b}, \textit{full-c}, which have different min/max bounds for the parameters of interest or prior covariances.
For the \textit{full-a} case we define the minimum \(L^i\) and maximum \(U^i\) bounds of each parameter 
\(x_i\in \{V_{\ell}, T_p, J_k\}\) as follows: let  $\xx_i$ be the maximum likelihood  value of the parameter $x_i$,  found approximately, by running an optimization algorithm on all the parameters;\footnote{The optimization algorithm used here is BFGS \cite{BFGS}, but in principle any other could be used.} we then take \(L^i = 0.1 \xx_i\), \(U^i = 10 \xx_i\), \(i=1,...,338\). 
Since the values of physical parameters \(x_i\) may differ by several orders of magnitude, it makes sense to apply a transform to get a similar magnitude for all the parameters.
Such a transform was done by function \(\log_{10}\) and a constant shift, mapping the parameters \(x_i\) from the original range \( [L^i, U^i]\)  to \([-1,1]\).
 The prior covariance is taken as  \(C_{\text{full-a}} = 0.25\cdot I_{338}\). 
 So, all the parameters in the transformed representation vary within the box \([-1, 1]\) and have standard deviation \(0.5\). For the \textit{full-b} case wider parameter bounds are taken: \(L^i = 0.001 \xx_i\), \(U^i = 1000 \xx_i\), \(i=1,...,338\). The parameters are transformed by \(\log_{10}\) function, and then mapped to the interval \([-3,3]\). The prior covariance is the same as in the \textit{full-a} case, so all the parameters have standard deviation 0.5 in the transformed representation. Case \textit{full-c} uses the same parameter bounds and the same transform as case \textit{full-b}, but a wider prior covariance \(C_{\text{full-c}} = 9\cdot C_{\text{full-a}}\), which means the prior standard deviation is 1.5 in the transformed representation.

\smallskip
$\bullet$ {\bf Lightweight parameterisation} Here we consider a reduced,  21-dimensional, parameter  space. 
Here we just  fix the values of $V_1,\dots,V_{86}$ (non-aquifer blocks), and we find the remaining $V_{87},\dots, V_{124}$ (aquifer blocks), $T_1,  \dots,T_{139}$ (all blocks), $J_1, \dots J_{75}$ (all perforations), which are required by the simulator, using  21 new parameters. Such parameters essentially act as multipliers;  namely,  for each one of the seven layers $n \in \{\text{A},\dots \text{G}\}$ we introduce one {\em pore volume multiplier for the aquifer blocks} \(\tilde{V}_n\), one {\em transmissibility multiplier} \(\tilde{T}_n\), and one perforation productivity multiplier \(\tilde{J}_n\). These parameters, collectively denoted by $y \in\R^{21}$, are those that we are interested in sampling from, by using the posterior measure 
$$
\pi_{light}(y) \propto 
{L}(d_0|X(y)) \cdot \rho(y),
$$
where $\rho(y)$ is a zero mean Gaussian with covariance matrix denoted by $C_{21}$, described below. Because we are using the same simulator as for the full reservoir simulation, the likelihood function $L$ is still the one defined in  \eqref{eq:ilya_2}, hence necessarily we must have $X(y)\in \R^{338}$. To define the function $X:\R^{21}\rightarrow \R^{338}$, we need to introduce some notation first.  Denote by $A_n$  the number of aquifer blocks in layer $n$, $P_n$ the number of transmissibility coefficients in layer $n$, and $K_n$ the number of well perforations in this layer. Let $\VV_{\ell}$ be the maximum likelihood value of the parameter $V_{\ell}$ (similarly for $\TT_{p}$ and $\JJ_{k}$), again found by running a maximum likelihood algorithm, and let the corresponding full vector denoted by $\xx$.
The first 86 components of $X(y)$ (corresponding to non-aquifer $V_{\ell}$) are taken equal to $\VV_{\ell}$, $\ell=1,\dots 86$, irrespective of the input $y$. The remaining $338-86=252$ components of $X(y)$ are found by a linear mapping $M\cdot y$, using a $252 \times 21$ sparse matrix $M$.
The first column of $M$ contains the vector 
  $$
  (\VV_{86+1}, \dots, \VV_{86+A_1}, 0 \dots\dots\dots\dots \dots \dots \dots\dots\dots 0)^T,
  $$
  the second column contains the vector
  \begin{align*}
   &(\underbrace{0, \dots\dots \dots \dots ,0},\VV_{86+L_1+1}, \dots, \VV_{86+A_1+A_2}, 0 \dots 0 )^T,   \\
   & \quad \quad \quad \quad L_1
  \end{align*}
 and so forth until the 7th column. The columns from 8 to 14 are built similarly, 
 such that column $n+7$ corresponds to layer $n$ and has $P_n$ non-zero values equal to $\TT_p$ (for appropriate indices $p$).
 The last seven columns are built in the same way, this time using the values $\JJ_{k}$.


For the \textit{lightweight} parameterisation the following minimum and maximum bounds were employed: 
\([0.15, 15]\) for all $\tilde{V}_n$, 
\([0.07, 7]\) for all $\tilde{T}_n$, and
\([0.11, 11]\) for all $\tilde{J}_n$. 
As before, the physical parameters (multipliers) $y_i$ are  mapped to the interval \([-1,1]\).
The prior covariance \(C_{21}\), which acts in the transformed variables, was taken as essentially a diagonal matrix with the main diagonal equal to \(0.25\), however additional non-zero covariances equal to \(0.1\) were also specified between the transmissibility multiplier \(\tilde{T}_n\) and perforation productivity multiplier \(\tilde{J}_n\) for each layer \(n\). A brief summary of the four discussed cases of the model parameterisation is presented in Table \ref{tab:ilya_table1}.

\begin{table}[h!]
  \setlength\extrarowheight{6pt}
  \begin{center}
    \begin{tabular}{|c|c|c|c|c|c|c|c|}
        
        \hline
        \multirow{2}{*}{Case} & \multirow{2}{*}{dim} & 
        \multirow{2}{*}{\parbox{2cm}{\centering parameters notation}} & \multicolumn{2}{c|}{\textit{for phys. par.}} & \multicolumn{3}{c|}{\textit{for transformed parameters}} \\[3pt]
        
        \cline{4-8}
        
        \multirow{2}{*}{} & \multirow{2}{*}{} & \multirow{2}{*}{} & \(l_i\) & \(u_i\) & params range & prior cov & prior std \\[3pt]
        \hline
        
        \textit{lightweight} & 21 & \(\tilde{V}_n, \tilde{T}_n, \tilde{J}_n\), or $y_i$ & 0.1 & 10 & \([-1,1]\) & \(C_{21}\approx \text{diag}\) & \(\approx 0.5\) \\[3pt]
        \hline

        \textit{full-a} & 338 & \(V_{\ell}, T_p, J_k\), or $x_i$ & 0.1 & 10 & \([-1,1]\) & \(C_\text{full-a}=\text{diag}\) & \(0.5\) \\[3pt]
        \hline
        
        \textit{full-b} & 338 & \(V_{\ell}, T_p, J_k\), or $x_i$ & 0.001 & 1000 & \([-3,3]\) & \(C_\text{full-b}=C_\text{full-a}\) & \(0.5\) \\[3pt]
        \hline

        \textit{full-c} & 338 & \(V_{\ell}, T_p, J_k\), or $x_i$ & 0.001 & 1000 & \([-3,3]\) & \(C_\text{full-c}=9 C_\text{full-a}\) & \(1.5\) \\[3pt]
        \hline

    \end{tabular}

    \vspace{3pt}

    \caption{Summary of the subcases for the reservoir simulation model. In physical representation the lower bounds are \(L^i=l_i b_i\), the upper bounds are \(U^i=u_i b_i\), where \(l_i, u_i\) are reported in the table, and \(b_i\) are some base case parameter values (e.g. for all \textit{full} parameterisations $b_i=\xx_i$).}
    \label{tab:ilya_table1}

  \end{center}
\end{table}

\section{Numerics: sampling from measures supported on bounded domains}\label{sec4}

To compare efficiency of the algorithms we compute a normalised effective sample size (nESS), where the normalisation is by the number of samples $N$. 
Following \cite{MA}, we define the Effective Sample Size $ESS = N/\tau_{int}$ where $N$ is the number of steps of the chain  (after appropriate burn-in) and $\tau_{int}$ is the integrated autocorrelation time, $\tau_{int}:= 1+ \sum_k \gamma (k)$, where $\gamma(k)$ is the $lag-k$ autocorrelation. Consistently, the normalised ESS, nESS, is just $nESS:= ESS/N$.  Notice that $nESS$ can be bigger than one (when the samples are negatively correlated), and this is something that will appear in our simulations. 
As an estimator for $\tau_{int}$  we will take the Bartlett window estimator (see for example \cite[Section 6]{MA}, and references therein)   
rather than the initial monotone sequence estimator (see again\cite[Section 6]{MA}), as the former is more suited to include non-reversible chains. 
Since the nESS is itself a random quantity, we performed 10 runs of each case using different seeds, and our plots below show P10, P50, P90 percentiles of the nESS from these runs. 


\subsection{Bounces vs Rejection}
First we consider the performance of the two proposed methods for sampling from the box $B$. We illustrate these by comparing 
SOL-HMC-bounce and SOL-HMC-rej.

In Figure \ref{fig:ilya_f1} we compare performance of
SOL-HMC-bounce and SOL-HMC-rej for sampling from the 5D Rosenbrock target $\pi_{Ros}$; each one of the five parameters is taken to vary in the interval  \([-a,a]\), and  Figure \ref{fig:ilya_f1} shows how the performance varies as the size $a$ of the box varies,  \(a =0.1,0.2,\ldots,1.4\).  The target acceptance rate for both samplers was set to 0.9, and parameter \(i=0.6\) (defined in \eqref{parami}).
\begin{figure}[h]
\noindent
\centering
\begin{subfigure}{.5\textwidth}
  \centering
  \includegraphics[width=\linewidth]{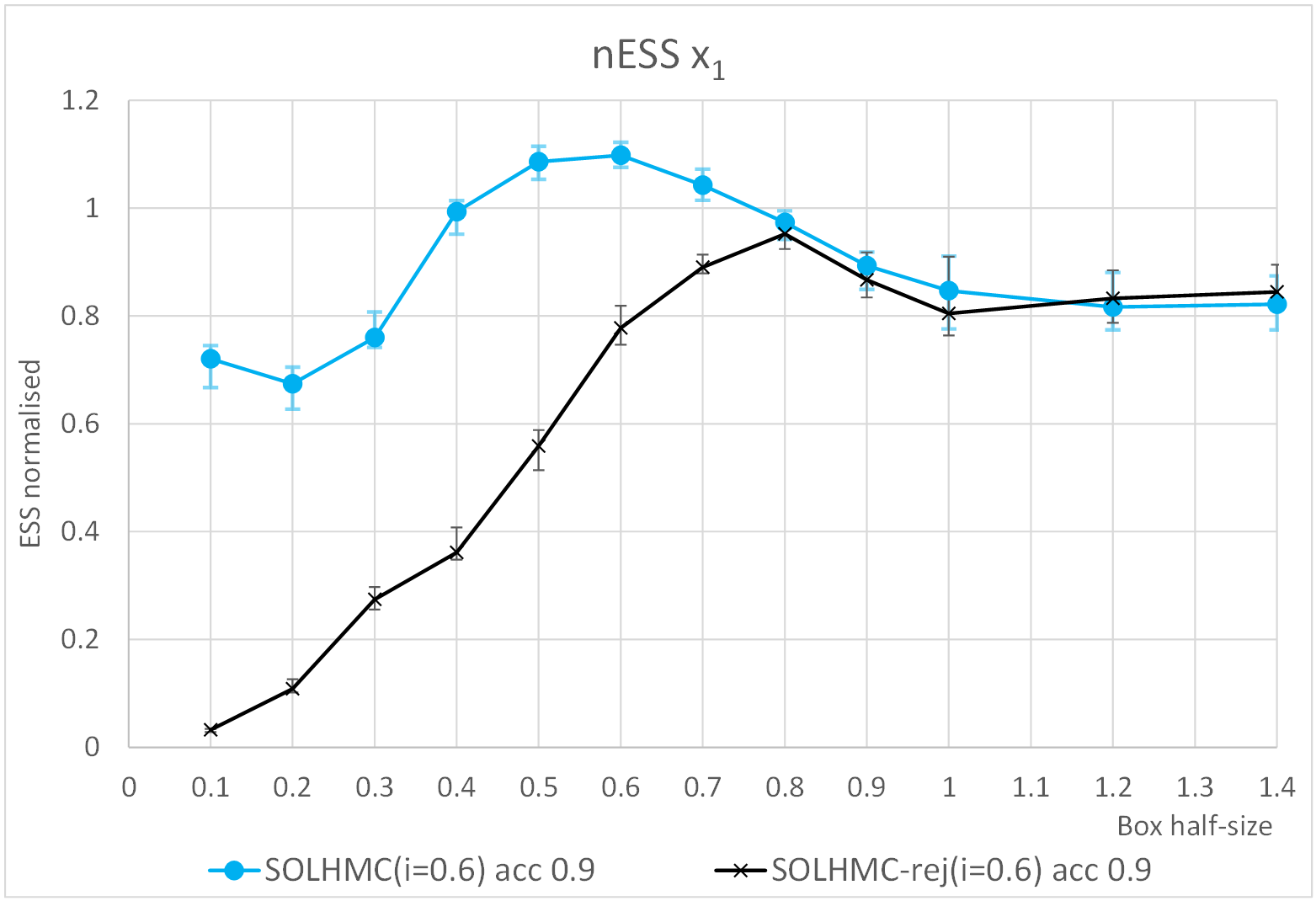}
\end{subfigure}%
\begin{subfigure}{.5\textwidth}
  \centering
  \includegraphics[width=\linewidth]{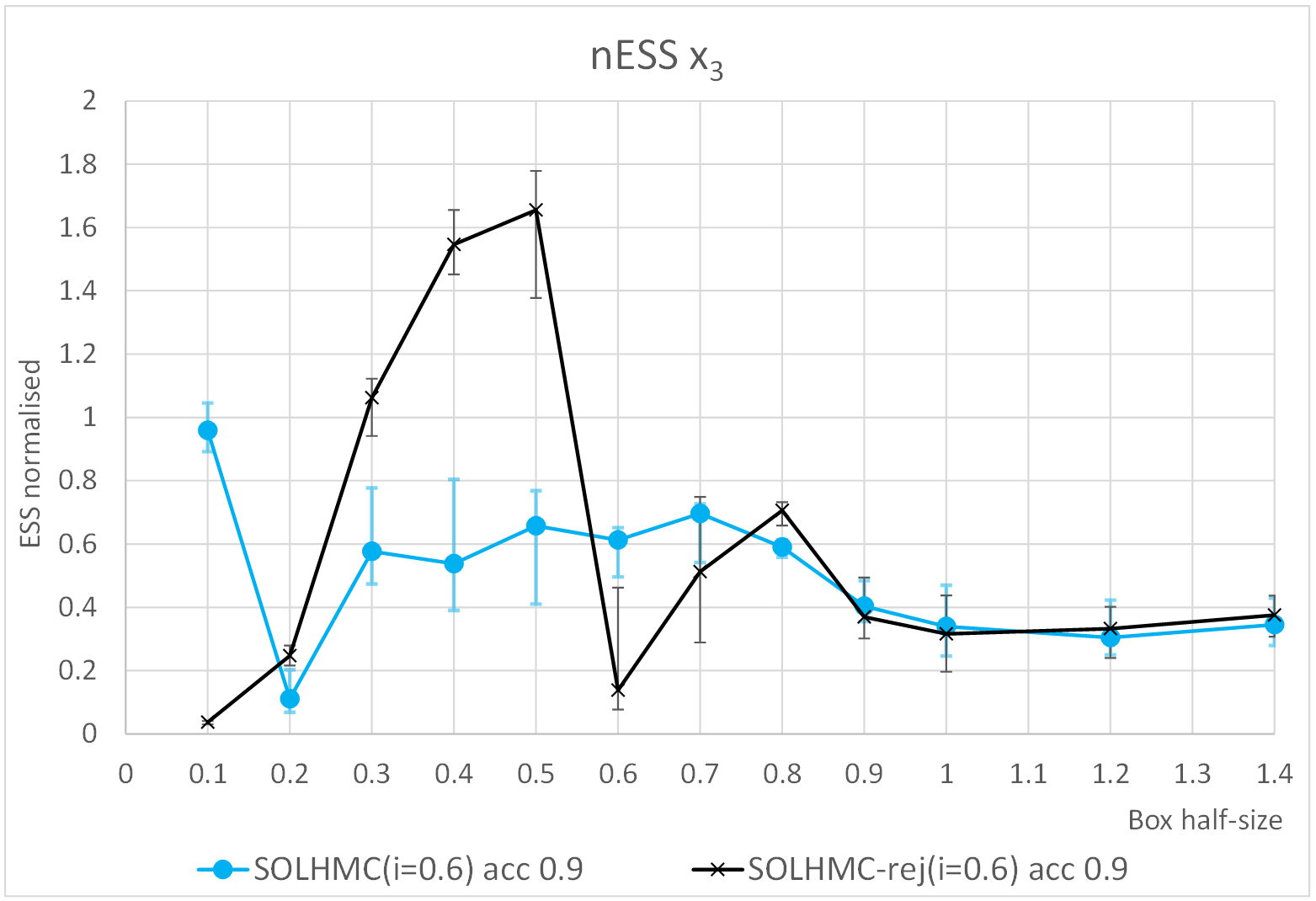}
\end{subfigure}
\caption{
Normalised ESS for SOL-HMC-bounce (\textit{blue}) and SOL-HMC-rej (\textit{black}), for different sizes of the box bounding the parameter space (\textit{X-axis}). The two plots correspond to coordinates \(x_1, x_3\) only. The other three coordinates have nESS plots similar to \(x_3\).
}
\label{fig:ilya_f1}
\end{figure}
%

As a ``sanity test" the plots indicate that for the larger boxes (\(a\geq 0.8\)) the 
two implementations SOL-HMC-bounce and SOL-HMC-rej are almost identical 
 (in terms of nESS), which is natural as for large box sizes these two algorithms coincide. 
 For small box sizes, the performance of the two samplers depends really on which coordinate is being sampled, so 
the performance of the two algorithms is substantially indistinguishable for this low dimensional problem.

It is important to note the following practical drawback of SOL-HMC-rej (or indeed any other sampler which handles boundaries by the rejection mechanism) with respect to SOL-HMC-bounce: during the proposal step a trajectory may leave the box, and then return back inside the box. By construction of the algorithm such a trajectory is not rejected just because it escaped from the domain for a short while. The accept/reject decision is made only for the final point of the proposal trajectory, and thus every trajectory needs to be calculated till the end.
However, if the trajectory is allowed to leave the box for the intermediate calculations, it may go to the extreme regions of the parameter space, where the simulator may suffer from severe numerical errors and abnormal behaviour. 
We illustrate this phenomena by comparing SOL-HMC-rej against HMC-bounce in Figure \ref{fig:ilya_f2} for \textit{full-a} in (A) and \textit{full-b} in (B). (Here we think of  HMC-bounce as sort of gold standard and for this reason we  compare SOL-HMC-rej with HMC-bounce). 
We examine the ratio of nESS of SOL-HMC-rej and HMC-bounce and plot a histogram for the parameters. When the nESS ratio is bigger than one then SOL-HMC-rej is performing better than HMC-bounce. This is the case in (B) for \textit{full-b}. However in (A) for \textit{full-a} the boundary of $B$ is encountered far more frequently, just because the size of the box for this target measure is smaller, see Table \ref{tab:ilya_table1}. Moreover a comparison of  the histograms in Figure \ref{fig:ilya_f2} (A) with the one in Figure \ref{fig:ilya_f7D} (A) shows better performance of SOL-HMC bounce with respect to SOL-HMC Rejections.
From now on we consider SOL-HMC-bounce only. 

\begin{figure}[h]
\noindent
\centering

\begin{subfigure}{.45\textwidth}
  \centering
  \includegraphics[width=\linewidth]{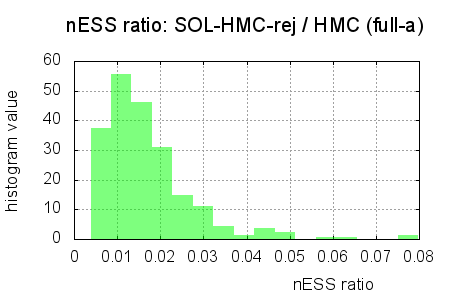}
  \vspace*{-12mm}
  \caption{}
\end{subfigure}%
\begin{subfigure}{.45\textwidth}
  \centering
  \includegraphics[width=\linewidth]{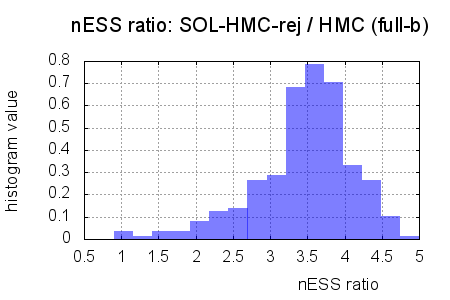}
  \vspace*{-12mm}
  \caption{}
\end{subfigure}
\caption{
Ratio of nESS (SOL-HMC-rej divided by HMC-bounce). Parameterisation \textit{full-a} (A) is shown in green and \textit{full-b} (B) in blue. 
}
\label{fig:ilya_f2}
\end{figure}

\subsection{Comparison for 5D Rosenbrock}
We consider the 5D Rosenbrock target $\pi_{Ros}$ where the minimum-to-maximum range for each one of the five parameters was taken as \([-a,a]\), where \(a =0.1,0.2,\ldots,1.4\).
The plots in Figure \ref{fig:ilya_f4} compare the performance of the HMC-bounce, SOL-HMC-bounce, and Horowitz-bounce algorithms.
The target acceptance rate is 0.9, and the parameter \(i=0.6\) for SOL-HMC-bounce and Horowitz-bounce. For this small dimensional problem  we observe that SOL-HMC-bounce and Horowitz-bounce have similar nESS across the range of sizes for the box $B$. For smaller boxes $B$ (e.g. $a\leq 0.5$) all three algorithms have similar nESS. For larger box sizes we see for parameter $x_1$ an advantage in using SOL-HMC-bounce/Horowitz-bounce
over the HMC-bounce however for $x_2$ there is a slight advantage to HMC-bounce. This corroborates the idea that in low dimension the advantage of introducing irreversibility in the sampler is hardly noticeable. 



\begin{figure}[h]
\noindent
\centering
\begin{subfigure}{.5\textwidth}
  \centering
  \includegraphics[width=\linewidth]{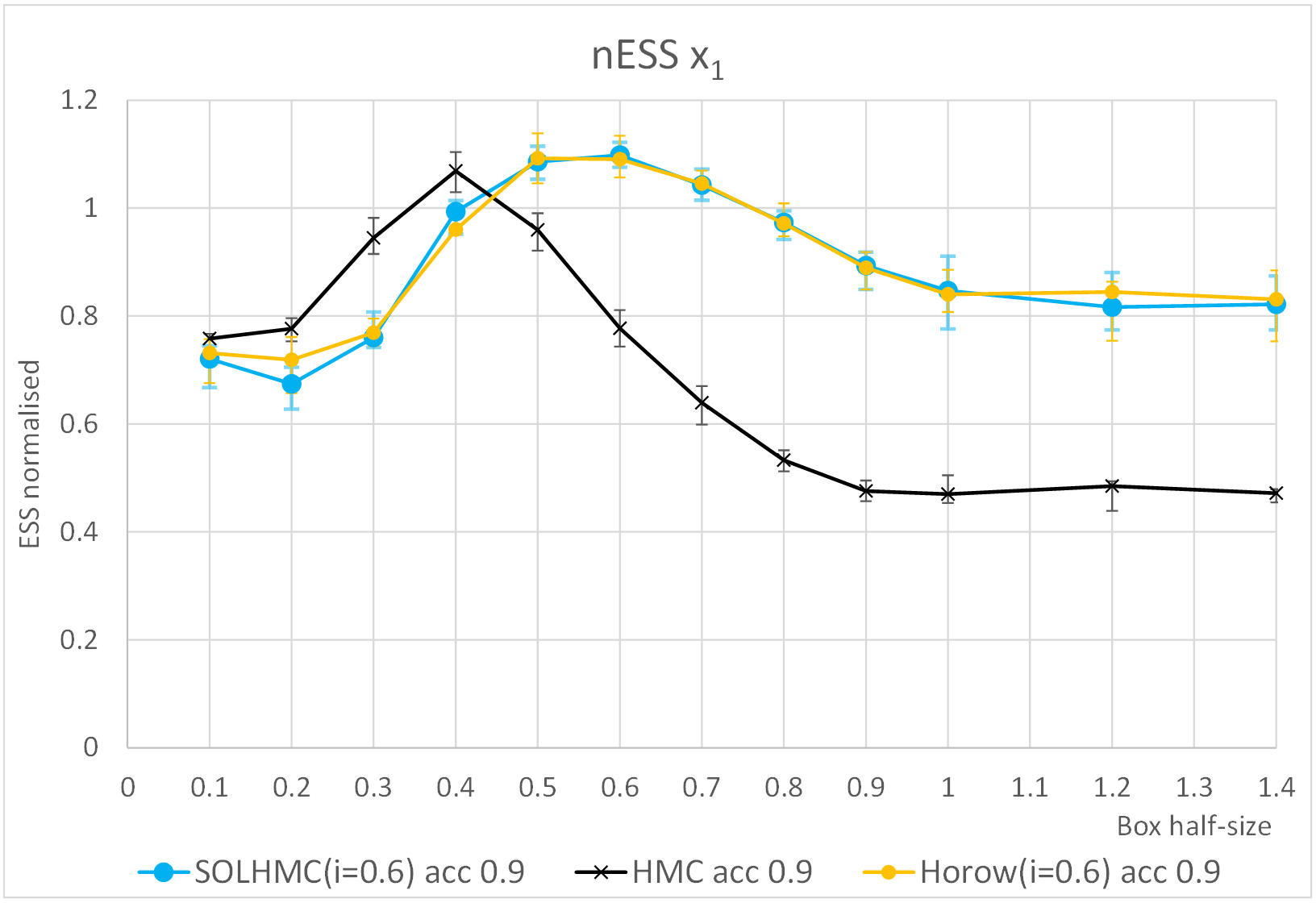}
\end{subfigure}%
\begin{subfigure}{.5\textwidth}
  \centering
  \includegraphics[width=\linewidth]{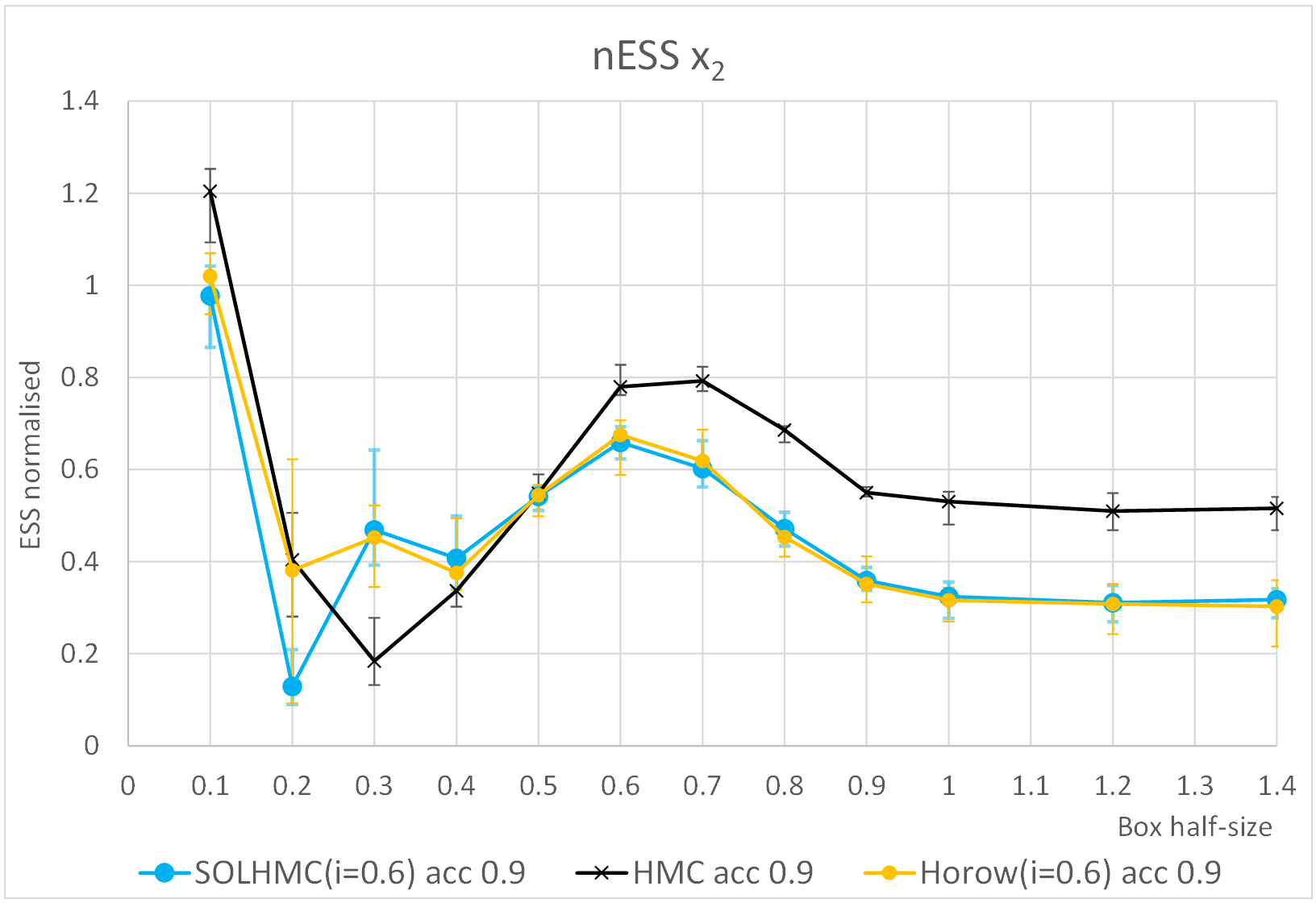}
\end{subfigure}
\caption{
Normalised ESS for SOL-HMC-bounce (\textit{blue}), HMC (\textit{black}), and Horowitz (\textit{orange}), for different sizes of the box bounding the parameter space (\textit{X-axis}). The two plots correspond to coordinates \(x_1, x_2\) only. The other three coordinates show a similar picture.
}
\label{fig:ilya_f4}
\end{figure}

\subsection{Increasing parameter space and effectiveness of non-reversibility.}

%

We now consider our more realistic targets and increase the parameter space to $21$ and then to $338$. We now clearly see the advantage of the non-reversible algorithms SOL-HMC-bounce and Horowitz-bounce over HMC-bounce.

Figure \ref{fig:ilya_f5} reports the nESS for the lightweight parameterisation of the reservoir simulation problem for the following four cases: HMC-bounce with an acceptance rate of 0.82 (target 0.8), SOL-HMC-bounce with acceptance rate of 0.77 (target 0.9), 
SOL-HMC-bounce with an acceptance rate of 0.69 (target 0.8)
and Horowitz-bounce with an acceptance rate of 0.68 (target 0.8).
All SOL-HMC-bounce and Horowitz-bounce algorithms took the parameter \(i=0.5\) and here we give results from a single MCMC run in each case.
The plot clearly shows that the non-reversible algorithms outperform HMC for the majority of the parameters. We also observe the variability due to acceptance rate: for SOL-HMC-bounce a better nESS is achieved for the higher acceptance rate.

\begin{figure}[h]
\noindent
\includegraphics[keepaspectratio=true, width=0.8\textwidth]{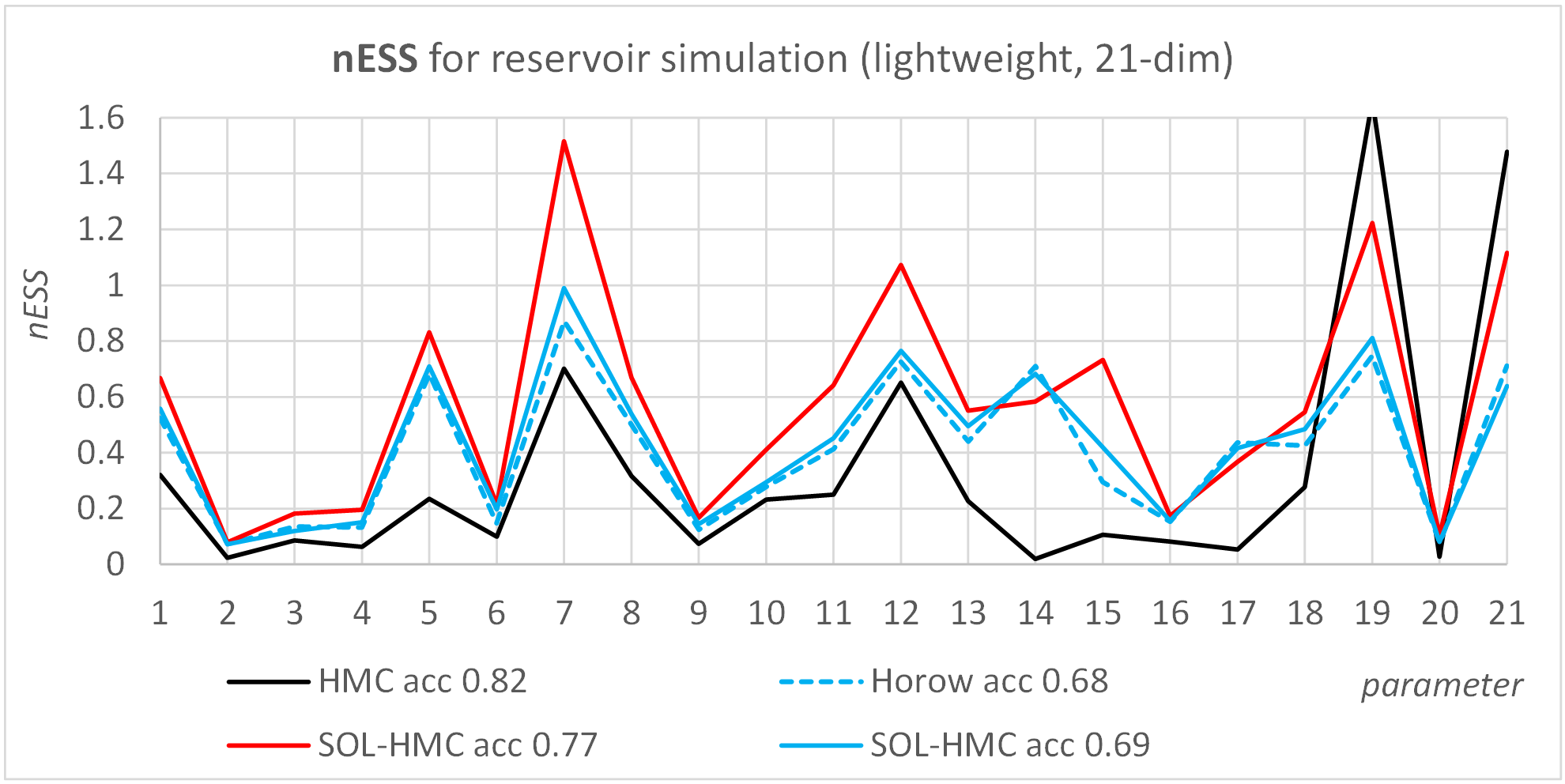}
\caption{Normalised ESS (\(Y\) axis) for the reservoir simulation MCMC, lightweight parameterisation. \(X\) axis shows the 21 parameters. In the legend, the real acceptance rates are indicated.}
\label{fig:ilya_f5}
\end{figure}

As we further increase the dimension and complexity the advantage of the non-reversible algorithm becomes further apparent. In Figure \ref{fig:ilya_f5A} we compare for \textit{full-a} SOL-HMC-bounce and HMC-bounce and observe a clear improved nESS for SOL-HMC-bounce across the whole parameter space.

%

\begin{figure}[h]
\noindent
\includegraphics[keepaspectratio=true, width=0.8\textwidth]{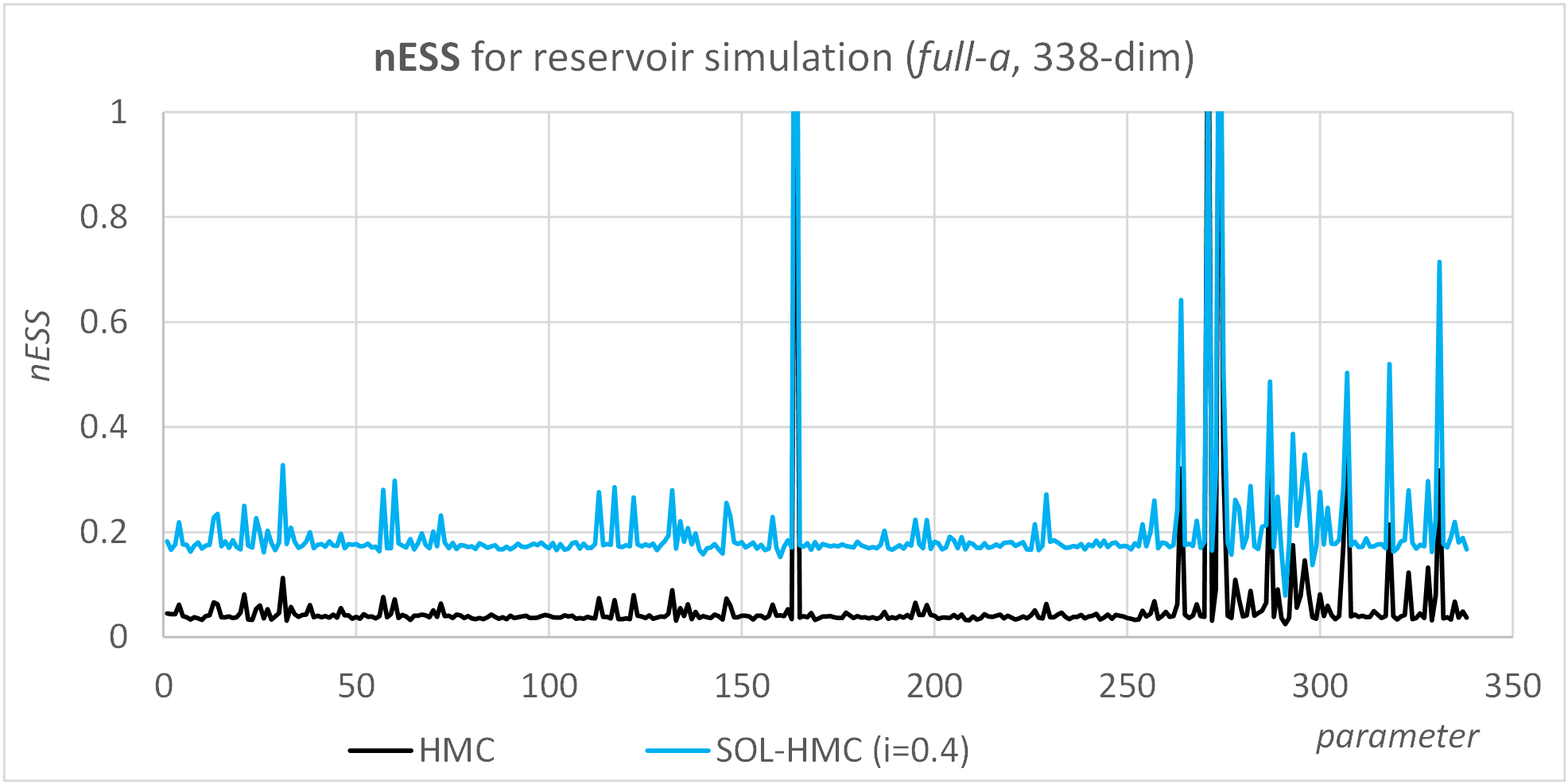}
\caption{Normalised ESS (\(Y\) axis) for the reservoir simulation MCMC, \textit{full-a} parameterisation. \(X\) axis shows the 338 parameters. The samplers are HMC and SOL-HMC with \(i=0.4\).}
\label{fig:ilya_f5A}
\end{figure}


Finally we compare SOL-HMC-bounce and Horowitz against the benchmark of HMC-bounce by examining the ratio of nESS. Recall that when the ratio is bigger than one then SOL-HMC-bounce (or Horowitz) has a larger nESS than HMC.
We consider the targets \textit{full-a}, \textit{full-b} and \textit{full-c}.
In Figure \ref{fig:ilya_f7B} we compare for \textit{full-a} SOL-HMC-bounce against Horowitz-bounce. First note that in both bases the nESS ratio is $>1$ for most parameters showing a clear improvement in the  non-reversible algorithms over HMC. To aid comparison between SOL-HMC-bounce against Horowitz-bounce we plot on (A) and (B) a fit of the histogram from (A), this is the black dotted line.
We see that the nESS for SOL-HMC-bounce over the parameters is larger than that for Horowitz-bounce and that there is an improvement using  
SOL-HMC-bounce. Here we took $i=0.5$. 
\begin{figure}[h!]
\noindent
\centering
\begin{subfigure}{.45\textwidth}
  \centering
  \includegraphics[width=\linewidth]{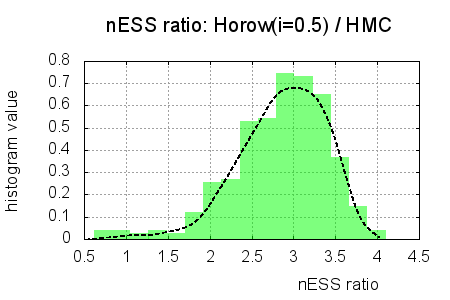}
  \vspace*{-12mm}
  \caption{}
\end{subfigure}%
\begin{subfigure}{.45\textwidth}
  \centering
  \includegraphics[width=\linewidth]{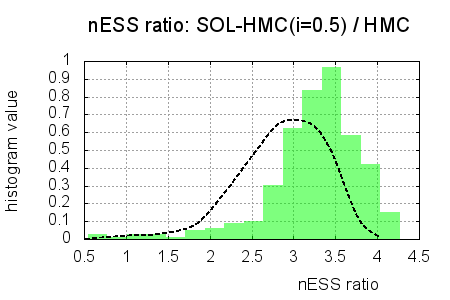}
  
  \vspace*{-12mm}
  \caption{}
\end{subfigure}
\caption{Ratio of nESS for Horowitz-bounce  by nESS for SOL-HMC-bounce. The target measure here is the 338-dimensional \textit{full-a}. Parameter $i=0.5$. 
}
\label{fig:ilya_f7B}
\end{figure}
Figure \ref{fig:ilya_f7C} examines the target \textit{full-b}. For both 
SOL-HMC-bounce against Horowitz-bounce we see an improvement over the reversible HMC algorithm as the ratios are $>1$ for all parameters. We also observe a shift to larger values and hence improvement in the nESS for SOL-HMC-bounce (B) compared to Horowitz-bounce (A). In this figure we took $i=0.7$. This can be compared to Figure \ref{fig:ilya_f7D} (B) where $i=0.4$.
\begin{figure}[h!]
\noindent
\centering
\begin{subfigure}{.45\textwidth}
  \centering
  \includegraphics[width=\linewidth]{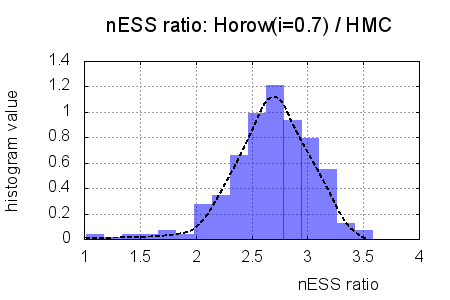}
  \vspace*{-12mm}
  \caption{}
\end{subfigure}%
\begin{subfigure}{.45\textwidth}
  \centering
  \includegraphics[width=\linewidth]{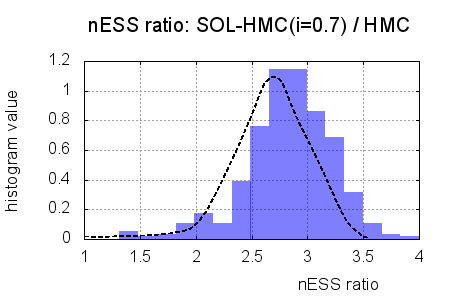}
  \vspace*{-12mm}
  \caption{}
\end{subfigure}
\caption{
Ratio of nESS for Horowtiz-bounce and SOL-HMC-bounce for target \textit{full-b} ($i=0.7$).
}
\label{fig:ilya_f7C}
\end{figure}

Finally, in Figure \ref{fig:ilya_f7D}, we examine SOL-HMC-bounce for \textit{full-a} (A), \textit{full-b} (B) and \textit{full-c} (C) for the same value of $i=0.4$. We see a clear improvement of the non-reversible SOL-HMC-bounce over HMC in each case.
We compare here to the SOL-HMC-bounce for \textit{full-b} for the same 
value of $i=0.4$ in (B). We observe a similar improvement for SOL-HMC-bounce over HMC in both cases.

\begin{figure}[h!]
\noindent
\centering
\begin{subfigure}{.32\textwidth}
  \centering
  \includegraphics[width=\linewidth]{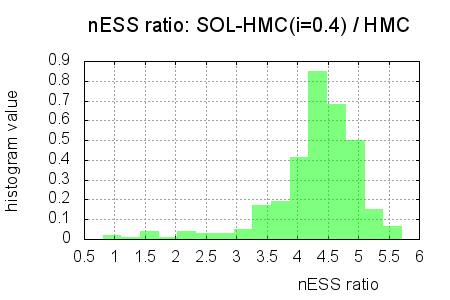}
  \vspace*{-12mm}
  \caption{}
\end{subfigure}%
\begin{subfigure}{.32\textwidth}
  \centering
  \includegraphics[width=\linewidth]{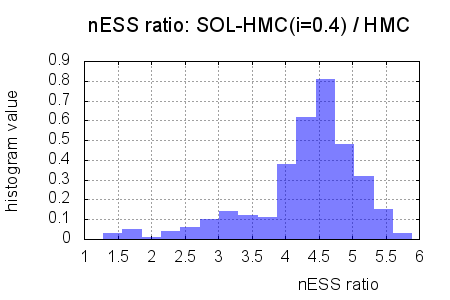}
  \vspace*{-12mm}
  \caption{}
\end{subfigure}%
\begin{subfigure}{.32\textwidth}
  \centering
  \includegraphics[width=\linewidth]{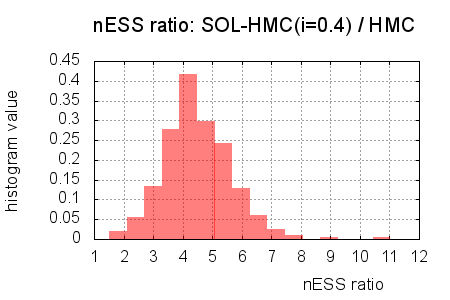}
  \vspace*{-12mm}
  \caption{}
\end{subfigure}
\caption{
Ratio of nESS for SOL-HMC-bounce for targets \textit{full-a} (A), \textit{full-b} (B) and \textit{full-c} (C) and in each case $i=0.4$.
}
\label{fig:ilya_f7D}
\end{figure}

\section{Conclusion}
We have investigated two different ways to deal with sampling measures on a bounded box $B$: rejection and bounces. This is crucial in many practical applications, for example to respect physical laws (such as porosity for reservoir modelling or pixel values in image reconstruction). We have explained and demonstrated why, for complex problems involving the use of a proxy, reflection algorithms should be preferred to rejection strategies. 
We have furthermore shown that when sampling from complex realistic target measures, such as those that arise in reservoir simulation,  non-reversible algorithms such as SOL-HMC and Horowitz outperform standard reversible algorithms such as HMC. In addition, we see that as the problem size grows SOL-HMC is superior to Horowitz having larger nESS.

\section*{Acknowledgements}
The work of I. Fursov and G. J. Lord was supported by the EPSRC EQUIP grant (EP/K034154/1).  P. Dobson was supported by the Maxwell Institute Graduate School in Analysis and its
Applications (MIGSAA), a Centre for Doctoral Training funded by the UK Engineering and Physical
Sciences Research Council (grant EP/L016508/01), the Scottish Funding Council, Heriot--Watt
University and the University of Edinburgh.

\section*{Appendix A}

This Appendix gathers some basic results about the SOL-HMC-bounce algorithm, presented in Section \ref{sec2}. Throughout we use the notation introduced in Section \ref{sec2}.  
\begin{prop}\label{prop:preservesinvmeas}
	The SOL-HMC-bounce algorithm with reflections preserves the target measure.
\end{prop}

\begin{proof}
	It is easy to see  that the operator $\mathcal{O}^\varepsilon$ preserves the target measure $\tilde{\pi}$. Indeed $\mathcal{O}^\varepsilon$ leaves the $x$-variable untouched so, because $\tilde{\pi}$ is the product of $\pi(x)$ and a standard Gaussian in the $p$ variable, looking at the definition \eqref{oe1}-\eqref{oe2}  of  $\mathcal{O}^\varepsilon$, all one needs to show is that if $p$ is drawn from a standard Gaussian  then $\hat{p}:=pe^{-\varepsilon} + i\xi$ is also a Gaussian random variable -- here $\xi$ is a standard Gaussian independent of $p$. This is readily see as, by definition,  $\hat{p}$ has expectation $0$ and variance $1$,  since $e^{-2\varepsilon}+i^2=1$. Therefore if $(x,p)$ are drawn from $\tilde{\pi}$ then $\mathcal{O}^\varepsilon(x,p)=(x,\hat{p})$ is also distributed according to $\tilde{\pi}$. 
		
		Let $\chi=\chi^\delta_{S,bounce}$ denote the integrator described in the SOL-HMC-bounce algorithm. Since $\mathcal{O}^{\varepsilon}$ preserves the target measure $\tilde{\pi}$ it remains to show that the combination of the integrator $\chi$ and the accept--reject mechanism preserves the target measure. 
		
		It is well known, for instance see \cite[Theorem 9]{Sanz}, that if the integrator $\chi^\delta_{S,bounce}$ is {\em reversible under momentum flip} (that is,  $\chi^\delta_{S,bounce}\circ S= S\circ (\chi^{\delta}_{S,bounce})^{-1}$
		where $S(x,p)=(x,-p)$) and volume preserving then the composition of  $\chi^\delta_{S,bounce}$ and of the  accept-reject move satisfies the detailed balance equation. In particular, this step also preserves the target measure $\tilde{\pi}$. 
		
		Therefore it is sufficient to show that $\chi^\delta_{S,bounce} = \Th^{\delta/2}\circ \Rotref \circ \Th^{\delta/2}$ is reversible under momentum flip  and volume preserving. Note that both $\Th^{\delta}$ and $\Rot^\delta$ are flows corresponding to a Hamiltonian system so they must be reversible and volume preserving, see \cite[Section 8.2.2 and 8.2.3]{Sanz}. The composition of these operators also has these two properties and including reflection preserves these two properties, therefore $\Rotref$ is volume preserving and reversible, and hence so is $\chi^\delta_{S,bounce}$. 
\end{proof}

\begin{prop} \label{prop:irreversible}
The SOL-HMC-bounce algorithm defined in Section \ref{sec2} is non-reversible.
\end{prop}

\begin{proof}[Proof of Proposition \ref{prop:irreversible}]
For simplicity we will only consider the case when $N=1, C=1$ and $V(x)=0$. That is, we consider the target measure to be the ``truncation of a standard two dimensional Gaussian", namely 
$$
\hat{\pi}(x,p) = \frac{1}{Z_a} e^{-\frac{1}{2}(x^2+y^2)} \mathbbm{1}_{[-a,a]}(x), 
$$
where $Z_{a}$ is a normalising constant. In this case $\Th^{\delta/2}$ is the identity map, and $\Rot^\delta$ can be written as
\be
\Rot^{\delta}(x,p) = (x\cos(\delta)+p\sin(\delta), p\cos(\delta)-x\sin(\delta)).
\ee
With these observations, if at time $k$ the chain is $(x^k,p^k)$ then we can write the proposed move $(\tilde{x}^{k+1},\tilde{p}^{k+1})$ in the $k+1$-th step of SOL-HMC-bounce as: 
$$
(\tilde{x}^{k+1},\tilde{p}^{k+1}) = \Rotref^\delta(x^k,p^ke^{-\varepsilon}+i\xi).
$$
where $\xi$ is drawn from a standard normal distribution.
In this case the acceptance probability is given by
$$
\alpha = \min(1, e^{-\frac{1}{2}((\tilde{x}^{k+1})^2+(\tilde{p}^{k+1})^2-(x^k)^2-(p^ke^{-\varepsilon}+i\xi)^2)}).
$$
Now we wish to calculate the transition kernel, $\K((x,p),(y,q))$, for this Markov chain, i.e. find the probability density corresponding to the move from $(x,p)$ to $(y,q)$.
 
Observe that $\Rot^\delta$ is a rotation about the origin and hence preserves radial distance, that is if $(\tilde{x},\tilde{p}):=\Rot^{\delta}(x,p)$  then
\be
	\tilde{x}^2+\tilde{p}^2=x^2+p^2.
\ee
Flipping momentum sign, i.e. applying reflections $S$,  also preserve radial distance, therefore the operator $\Rotref^\delta$ preserves radial distance. In particular, if $\tilde{x}^2+\tilde{p}^2< a^2$ (or equivalently $x^2+p^2<a^2$) then $\Rot^\delta(x,p)$ must remain in the strip $[-a,a]\times \R$, so in this situation $\Rotref^\delta(x,p)=\Rot^\delta(x,p)$. 

Suppose that $y^2+q^2\leq a^2$. Fix some $x\in[-a,a],p\in\R$, then let $(\hat{x},\hat{p})=\mathcal{O}^\varepsilon(x,p)=(x,pe^{-\varepsilon}+i\xi)$, where $\xi$ is a standard normal random variable. In which case we have that $\hat{p}$ is normally distributed with mean $pe^{-\varepsilon}$ and variance $i^2$. Set $(y,q)=\Rot^\delta(\hat{x},\hat{p})$, then
\be
(y,q)=(\hat{x}\cos(\delta)+\hat{p}\sin(\delta), \hat{p}\cos(\delta)-\hat{x}\sin(\delta)).
\ee
Therefore $y$ is normally distributed with mean $x\cos(\delta)+pe^{-\varepsilon}\sin(\delta)$ and variance $i^2\sin(\delta)^2$. Once $y$ has been determined we may solve for $q$ and find
\be\label{eq:forwardcond}
q=\frac{y\cos(\delta)-x}{\sin(\delta)}.
\ee
In which case the transition kernel is given by
\begin{align*}
\K((x,p),(y,q)) &= \frac{1}{\sqrt{2\pi i^2 \sin(\delta)^2}} e^{-\frac{(y-x\cos(\delta)-pe^{-\varepsilon}\sin(\delta))^2}{2i^2\sin(\delta)^2}}\alpha \delta_\frac{y\cos(\delta)-x}{\sin(\delta)}(q)\\
& +(1-\alpha)\delta_x(y)\frac{1}{\sqrt{2\pi i^2}}e^{-\frac{(q+pe^{-\varepsilon})^2}{2i^2}}
\end{align*}
where $\alpha$ is the acceptance probability and is given by
\be
\alpha = \min(1,e^{\frac{1}{2}(x^2+p^2-y^2-q^2)}).
\ee

Now the algorithm is reversible if and only if the detailed balance condition holds, that is
\begin{equation}\label{eq:DB}
\hat{\pi}(x,p)\K((x,p),(y,p)) = \hat{\pi}(y,q)\K((y,q),(x,p)), \quad \forall x,y \in[-a,a], p,q\in\R.
\end{equation} 
To see that this does not hold consider the point $(x,p)=(0,0)$ and let $(y,q)$ be some point in the ball of radius $a$. Then by \eqref{eq:forwardcond} we must have
$y=q\tan(\delta)$, and the left hand side of \eqref{eq:DB} becomes 
\begin{align*}
\hat{\pi}(0,0)\K((0,0),(q\tan(\delta),q)) = \frac{1}{\sqrt{2\pi}}\frac{1}{\sqrt{2\pi i^2 \sin(\delta)^2}} e^{-\frac{(q\tan(\delta))^2}{2i^2\sin(\delta)^2}}\min(1,e^{-\frac{1}{2}(q^2+q^2\tan(\delta)^2)}) >0.
\end{align*}
On the other hand, if we suppose $0<\delta<\pi/4$ then to move from $(q\tan(\delta),q)$ to $(0,0)$ is not possible unless $q=0$, since \eqref{eq:forwardcond} in this case becomes
$
q\tan(\delta)=0
$.
Therefore for any $q\neq0$ the right hand side of \eqref{eq:DB} must be zero, in particular we have that the algorithm is not reversible.
\end{proof}

\section*{Appendix B: Description of Reservoir Model and Simulator}
The simulator we use is an \textit{in-house} single phase simulator working on an unstructured grid with finite volumes spatial discretization and backward Euler time discretization, calculating the dynamics of pressures and fluid flows in the subsurface porous media. 

To obtain the observed pressure data, a fine grid three-phase model was run in the first place, using Schlumberger Eclipse black oil simulator \cite{Manual}. 
The resulting output Eclipse pressures were perturbed by the uncorrelated Gaussian noise, with standard deviation \(\sigma_{BHP}\) = 20 bar for the well BHP data, and \(\sigma_{b}\) = 3 bar for the reservoir (block) pressure data. Altogether 380 measurement points were considered (365 for the BHP, 15 for the reservoir pressure), taken with time step of 6 months.
The data errors covariance matrix \(C_d\) is diagonal, with the entries equal to either \(\sigma_{BHP}^2\) or \(\sigma_{b}^2\).

In the forward simulation mode, the reservoir properties are fixed, the producing and injecting wells (indexed by $w$) are controlled by the volumetric flow rates $q_w$, and the output modelled data are the time-dependent pressures at the blocks $P_{\ell}$ and the bottom-hole pressures at the wells $P_w^{\text{BHP}}$. The equations describing the fluid flow are as follows. First, the volumetric flow rate $Q_{\ell j}$ between the pair of connected blocks $\ell,j$ is proportional to the pressure difference between them, which can be regarded as Darcy's law:

\begin{align}
Q_{\ell j} = T_{\ell j}(P_{\ell} - P_j - \rho \, g\, h_{\ell j}),
\label{eq:ilya_sim1}
\end{align}
where $\rho$ is the known liquid density, $h_{\ell j}$ is the known depth difference between the block centers, and $g$ is the acceleration due to gravity.

The inflow $q_{w \ell}$ into the perforation of well $w$ in block $\ell$ is proportional to the difference of the bottom-hole pressure (BHP) and the block pressure:

\begin{align}
q_{w \ell} = J_{w \ell}(P_{\ell} - P_w^{\text{BHP}} - \rho \, g \, h_w^{\ell}),
\label{eq:ilya_sim2}
\end{align} 
where $h_w^{\ell}$ is the depth difference between the block center and the BHP reference depth. The total inflow into well $w$ is obtained by summing up contributions related to this well; that is, 
\begin{align}
q_{w} = \sum_{\ell} q_{w \ell}.
\label{eq:ilya_sim3}
\end{align} 

Finally, the volumetric inflows and outflows for block $\ell$ are balanced, with the excessive/deficient fluid volume leading to the block pressure change via the following compressibility equation:
\begin{align}
c_{\ell} \, V_{\ell} \, \frac{\partial P_{\ell}}{\partial t} = \sum_j Q_{j \ell} - \sum_w q_{w \ell},
\label{eq:ilya_sim4}
\end{align} 
where $t$ denotes time, and the first (second, respectively) summation on the right hand side is taken over all the blocks $j$ connected to the block $\ell$  (all the wells $w$ perforated in block $\ell$, respectively). The compressibility $c_{\ell}$ of the block is supposed to be known. The simulated reservoir time spans 12 years.


\end{document}